\documentclass[11pt]{article}
\usepackage[fontset=fandol]{ctex}
\usepackage[margin=1in]{geometry}
\usepackage{amsmath,amssymb,amsthm}
\usepackage[backend=biber,style=authoryear,maxcitenames=2,maxbibnames=99]{biblatex}

\newcommand{\mytitle}{Logics of Filter Bubbles}

\newcommand{\myabstract}{%

Filter bubbles can be understood as configurations characterized by within-group proximity and separation from outsiders. Personalized feeds may form or preserve such configurations. We develop hybrid-style logics for reasoning about filter bubbles. The static logic characterizes bubble-shaped configurations, while its dynamic extension captures their formation and persistence under personalized feeds. We provide sound and strongly complete axiomatizations for the static and dynamic logics and prove their decidability.
}
\newcommand{\mykeywords}{filter bubbles $\cdot$ social network logic $\cdot$ hybrid logic $\cdot$ event models}

\usepackage{mathtools}
\usepackage{microtype}
\usepackage{tikz}
\usetikzlibrary{arrows.meta,fit}
\usepackage[pdfborder=0]{hyperref}
\hypersetup{
  pdftitle={Logics of Filter Bubbles},
  pdfauthor={Lei Li and Jialiang Yan},
  pdfkeywords={filter bubbles, social network logic, events models}
}

\allowdisplaybreaks
\numberwithin{equation}{section}

\makeatletter
\renewcommand\section{\@startsection{section}{1}{\z@}%
  {-2.0ex \@plus -.5ex \@minus -.2ex}%
  {.8ex \@plus .2ex}%
  {\normalfont\large\bfseries}}
\renewcommand\subsection{\@startsection{subsection}{2}{\z@}%
  {-1.5ex \@plus -.4ex \@minus -.2ex}%
  {.5ex \@plus .1ex}%
  {\normalfont\normalsize\bfseries}}
\newtheoremstyle{mythm}{1ex}{1ex}{}{}%
  {\bfseries}{}{.6em}{}
\theoremstyle{mythm}
\newtheorem{theorem}{Theorem}[section]
\newtheorem{lemma}[theorem]{Lemma}
\newtheorem{proposition}[theorem]{Proposition}
\newtheorem{corollary}[theorem]{Corollary}
\theoremstyle{definition}
\newtheorem{definition}[theorem]{Definition}
\renewenvironment{proof}[1][Proof]{\par
  \pushQED{\hfill$\square$}%
  \normalfont \topsep3\p@\@plus2\p@\@minus1\p@ \labelsep1em\relax
  \trivlist
  \item[\hskip\labelsep\bfseries #1]\ignorespaces
}{%
  \popQED\endtrivlist\@endpefalse
}
\makeatother

\AtBeginDocument{%
  \setlength{\abovedisplayskip}{5pt plus 2pt minus 2pt}%
  \setlength{\belowdisplayskip}{5pt plus 2pt minus 2pt}%
  \setlength{\abovedisplayshortskip}{3pt plus 1pt minus 1pt}%
  \setlength{\belowdisplayshortskip}{3pt plus 1pt minus 1pt}%
  \setlength{\jot}{1.5pt}%
  \setlength{\textfloatsep}{8pt plus 2pt minus 2pt}%
  \setlength{\floatsep}{6pt plus 2pt minus 2pt}%
  \setlength{\intextsep}{6pt plus 2pt minus 2pt}%
}

\newcommand{\Lang}{\mathcal{L}}
\newcommand{\Act}{\mathcal{F}}
\newcommand{\pre}{\operatorname{pre}}
\newcommand{\post}{\operatorname{post}}
\newcommand{\Prof}{\mathbb{P}}
\newcommand{\Sim}{\mathbb{T}}
\newcommand{\FBForm}{\mathbb{F}}
\newcommand{\FBPres}{\mathbb{R}}
\newcommand{\FBDiss}{\mathbb{D}}

\begin{document}
\title{\mytitle}
\author{Lei Li\thanks{School of Philosophy, Shaanxi Normal University, Xi'an, China. Email: \texttt{lileity@snnu.edu.cn}} \and
Jialiang Yan\thanks{Institute of Logic, China University of Political Science and Law, Beijing, China. Email: \texttt{jialiangyan@cupl.edu.cn}}}
\date{}
\maketitle
\begin{abstract}
\myabstract
\end{abstract}
\bigskip

\section{Introduction}

The term \emph{filter bubble} is typically used both for a social configuration and for an account of how that configuration arose. Structurally, agents within a group are close to one another and separated from outsiders. Dynamically, that configuration forms or persists under a personalized feed \parencite{zuiderveenborgesius2016,geschke2019}. In this paper, we call the configuration itself a \emph{bubble} and call \emph{filter bubble} for a bubble that forms or persists under a specified feed update. We aim to develop logics to represent the configuration and the update.

Logical models of social influence begin with interpersonal relations. Beliefs or links change through friendship, communication, or observation \parencite{seligman2013facebook,christoff2015,pedersen2019,baccini2025,wangyi21}. We call this an \emph{internal} perspective on social change. Platform curation introduces a different kind of intervention. A recommender may send different feeds to recipients who need not communicate with, influence, or even know one another. Logics of personalized announcements provide a related account of directed informational change \parencite{belardinelli2025personalized}. Our \emph{external} perspective instead studies how a platform's allocation changes attitude profiles and, through them, the pattern of distances across a population.

To formalize both a bubble-shaped configuration and the external feed process under which it forms or persists, we develop a network logic and its dynamic extension. The static language treats agents as named points in a bounded distance space and defines a bubble by bounded internal disagreement and stricter separation from outsiders. The dynamic extension adds deterministic feed models whose preconditions determine which agents undergo attitude-relevant updates and whose postconditions determine the resulting attitudes. The updated profiles then determine new threshold relations, so the
dynamic language can express bubble formation, persistence, and dissolution. In the finite setting, we provide sound and strongly complete axiomatizations for both systems and prove their decidability.

The rest of the paper is organized as follows. Section~2 introduces the static network models and the axiomatization. Sections~3 and 4 extend this framework to a dynamic setting by defining feed models, and prove the soundness and completeness of the resulting logic. Section~5 concludes.

\section{The static logic of filter bubbles}

\subsection{Static network models and attitude profiles}

Let $m$ be a positive integer and we define the finite threshold grid as follows:
\begin{equation}
G=\left\{0,\frac1m,\frac2m,\ldots,1\right\}.
\label{eq:grid}
\end{equation}
For $r,s\in G$, write $r\oplus s=\min\{1,r+s\}$. 


Let $A$ be a finite nonempty set of agents, $P$ be a nonempty set of propositional atoms, and $\Omega$ be a finite set of nominals, with $P\cap\Omega=\varnothing$ and $|\Omega|=|A|$. We define a static network model.

\begin{definition}[Static network model]
A static network model is a tuple
\[
M=\bigl(A,(\sim_r)_{r\in G},V\bigr),
\]
where $A$ is a finite nonempty set of agents; $V:P\cup\Omega\to\wp(A)$ is a valuation. For every $i\in\Omega$, $V(i)$ is a
singleton, and the map sending $i$ to the unique member of $V(i)$ is a
bijection from $\Omega$ onto $A$;  each $\sim_r\subseteq A\times A$ is a threshold relation.
For all $a,b,c\in A$ and $r,s\in G$, the threshold relations satisfy:
\[
\begin{gathered}
a\sim_r a\; \text{(reflexivity)},\qquad
a\sim_r b\Longleftrightarrow b\sim_r a\; \text{(symmetry)},\qquad
a\sim_1 b\; \text{(boundedness)},\\
a\sim_r b\Longrightarrow a\sim_s b\quad\text{if }r\leq s\; \text{(threshold nesting)}, \\
a\sim_r b\text{ and }b\sim_s c
\Longrightarrow a\sim_{r\oplus s}c\; \text{(triangle closure)}.
\end{gathered}
\]

\end{definition}

Every agent is named by exactly one nominal. Agents in $A$ and
their names in $\Omega$ nevertheless remain semantic and syntactic
objects, respectively. The relation $a\sim_r b$ means that $a$ and $b$ are at distance at most $r$. Based on the model, we define the grid-valued distance

\begin{equation}
d_M(a,b)=\min\{r\in G:a\sim_r b\}.
\label{eq:distance}
\end{equation}

Distinct agents may have distance zero. Conversely, every distance on $A$ with values in $G$ and diameter at most one induces the displayed family by $a\sim_r b$ exactly when $d_M(a,b)\leq r$\footnote{Distance-indexed modalities and relational representations of metric
spaces have been studied in logics of metric spaces
\parencite{kutz2003,wolterzakharyaschev2005}.}.

\subsection{Language and semantics}

The static language $\Lang_0$ is the least set of formulas generated by
\begin{equation}
\varphi ::= p\mid i\mid\neg\varphi\mid(\varphi\land\psi)
\mid\langle r\rangle\varphi\mid @_i\varphi,
\label{eq:staticgrammar}
\end{equation}
where $p\in P$, $i\in\Omega$, and $r\in G$. Each proposition letter $p\in P$ represents an agent's attitude. Other Boolean connectives
and the constants $\top$ and $\bot$ are defined as usual. A nominal
$i$ names a particular agent. The threshold
formula $\langle r\rangle\varphi$ says that some agent within 
threshold $r$ satisfies $\varphi$. $@_i\varphi$ is the hybrid formula \parencite{blackburn1995,blackburn2000}, and says that $\varphi$
holds at the agent named by $i$. As usual,
$[r]\varphi:=\neg\langle r\rangle\neg\varphi$ means that every agent
within threshold $r$ satisfies $\varphi$.

Given a static network model $M=\bigl(A,(\sim_r)_{r\in G},V\bigr)$, $a\in A$ is the agent relative to which the current formula is interpreted. The truth conditions are defined recursively as follows:

\begin{tabular}{@{}lcl@{}}
$M,a\models p$
& $\text{iff}$
& $a\in V(p)$ \\

$M,a\models i$
& $\text{iff}$
& $a\in V(i)$ \\

$M,a\models\neg\varphi$
& $\text{iff}$
& $M,a\not\models\varphi$ \\

$M,a\models\varphi\land\psi$
& $\text{iff}$
& $M,a\models\varphi$ and $M,a\models\psi$ \\

$M,a\models\langle r\rangle\varphi$
& $\text{iff}$
& $M,b\models\varphi$ for some $b\in A$ such that $a\sim_r b$ \\

$M,a\models @_i\varphi$
& $\text{iff}$
& $M,b\models\varphi$, where $V(i)=\{b\}$
\end{tabular}

\vspace{0.3cm}

The propositional clause means the agent's attitudes. The nominal clause makes $i$ true only at the agent it names. The threshold modality $\langle r\rangle$ quantifies over agents within threshold $r$ of the current agent. The hybrid operator $@_i$ shifts evaluation to the agent named by $i$, independently of the current evaluation point. Consequently, if $V(i)=\{b\}$ and $V(j)=\{c\}$, then
\[
M,a\models @_i\langle r\rangle j
\quad\text{iff}\quad
b\sim_r c.
\]

Thus $@_i\langle r\rangle j$ states that the agents named by $i$ and $j$ are within distance $r$. Its truth does not depend on the initial evaluation point $a$. 

For a set $\Gamma$ of formulas, we write $M,a\models\Gamma$ when
$M,a\models\gamma$ for every $\gamma\in\Gamma$. Local semantic consequence
is defined by
\[
\Gamma\models\varphi
\quad\text{iff}\quad
M,a\models\Gamma\text{ implies }M,a\models\varphi
\]
for every static network model $M$ and every $a\in A$. We also write
$M\models\chi$ when $M,a\models\chi$ for every $a\in A$, and
$\models\chi$ when $M\models\chi$ for every static network model $M$.

\subsection{Bubbles in static network}

We now define a group by agents' names. Let $M$ be a static network model and let $I$ be a nonempty  subset of $\Omega$. The group of agents named by $I$ in $M$ is
\(
F_I^M:=\bigcup_{i\in I}V(i),
\)
and the formula naming this group is
\(
f_I:=\bigvee_{i\in I}i.
\)
Since the naming map is a bijection, $F_I^M$ is a nonempty  subset of $A$.

Given $\varepsilon,\eta\in G$ with $\varepsilon<\eta$, call $F_I^M$ an $(\varepsilon,\eta)$-\emph{bubble} when
\begin{align}
d_M(a,b)&\leq\varepsilon
&& (\forall a,b\in F_I^M),
\tag{Cohesion}\label{eq:internal-cohesion}\\
d_M(a,c)&>\eta
&& (\forall a\in F_I^M,\ \forall c\in A\setminus F_I^M).
\tag{Separation}\label{eq:external-separation}
\end{align}

This definition draws on a structural idea of high within-group homogeneity together with high between-group heterogeneity \parencite{estebanray1994,bramson2016}.

The cohesion condition says that disagreement among members is at most
the lower threshold. It describes a like-minded cluster but
says nothing about its boundary. External separation adds that boundary
by requiring even the closest outsider to lie beyond the higher
threshold. The strict gap $\varepsilon<\eta$ prevents a group from
qualifying merely because of a knife-edge classification.

Formula~\eqref{eq:bubble}  defines an internally cohesive and externally isolated region.

\begin{equation}
B_I^{\varepsilon,\eta}:=
\bigwedge_{i\in I}@_i
(
\bigwedge_{j\in I}\langle\varepsilon\rangle j
\ \land\ [\eta]f_I
),
\label{eq:bubble}
\end{equation}

\noindent where the $\langle\varepsilon\rangle$-conjunct bounds the internal diameter of $F_I^M$. The $[\eta]$-conjunct says that each agent within  threshold $\eta$ is  named by a nominal in $I$.

\begin{proposition}
\label{prop:modal-definability}
Let $M$ be a static network model, we have
\[
M\models B_I^{\varepsilon,\eta} \;\text{iff}\; F_I^M \text{ is an} (\varepsilon,\eta)\text{-bubble}.
\]

\end{proposition}

\begin{proof}
 If $V(i)=\{a\}$ and $V(j)=\{b\}$, then $@_i\langle\varepsilon\rangle j$ holds exactly when $a\sim_\varepsilon b$. Moreover, $@_i[\eta]f_I$ holds exactly when every $\eta$-neighbour of the unique agent in $V(i)$ satisfies one of the nominals in $I$. Because the nominal valuation induces a bijection from $\Omega$ onto $A$, this is equivalent to the absence of an $\eta$-edge from an agent in $F_I^M$ to an agent in $A\setminus F_I^M$.
\end{proof}

\subsection{Axiomatic system $H_0$}

We axiomatize static network models by a Hilbert-style system $H_0$ containing propositional tautologies, modus ponens, 
and the replacement rule: if $\vdash\varphi\leftrightarrow\psi$, so is $\vdash C[\varphi]\leftrightarrow C[\psi]$,
where $C[\cdot]$ is any one-place formula context.
Necessitation rule for each $r\in G$, together with the following schemes:

\begin{center}
\begin{tabular}{@{}ll@{}}
\hline
\multicolumn{2}{@{}l}{\textit{Threshold schemes}} \\[2pt]

$\displaystyle [r]\varphi\to\varphi$
& $(T_r)$ \\

$\displaystyle \varphi\to[r]\langle r\rangle\varphi$
& $(B_r)$ \\

$\displaystyle [s]\varphi\to[r]\varphi
\qquad(r\leq s)$
& $(\mathrm{Mon})$ \\

$\displaystyle [r\oplus s]\varphi\to[r][s]\varphi$
& $(\mathrm{Tri})$ \\

$\displaystyle @_i\langle1\rangle j$
& $(\mathrm{Bd})$ \\[4pt]
\hline
\multicolumn{2}{@{}l}{\textit{Naming schemes}} \\[2pt]

$\displaystyle \bigvee_{i\in\Omega}i$
& $(\mathrm{Name})$ \\

$\displaystyle
@_i\langle r\rangle\varphi
\leftrightarrow
\bigvee_{j\in\Omega}
\bigl(
@_i\langle r\rangle j
\land @_j\varphi
\bigr)$
& $(\mathrm{Paste}_r)$ \\[4pt]
\hline
\multicolumn{2}{@{}l}{\textit{Hybrid schemes}} \\[2pt]

$\displaystyle
@_i(\varphi\to\psi)
\to
(@_i\varphi\to@_i\psi)$
& $(K_{@})$ \\

$\displaystyle
\neg @_i\varphi
\leftrightarrow @_i\neg\varphi$
& $(\mathrm{Self})$ \\

$\displaystyle @_i i$
& $(\mathrm{Ref}_{@})$ \\

$\displaystyle
@_i@_j\varphi
\leftrightarrow @_j\varphi$
& $(\mathrm{Agree})$ \\

$\displaystyle
@_i(\varphi\land\psi)
\leftrightarrow
(@_i\varphi\land @_i\psi)$
& $(\mathrm{Conj}_{@})$ \\

$\displaystyle
i\to(\varphi\leftrightarrow @_i\varphi)$
& $(\mathrm{Locality})$ \\

$\displaystyle @_i\neg j
\qquad(i\neq j)$
& $(\mathrm{Dist})$ \\
\hline
\end{tabular}
\end{center}

 In addition, we have the rule $(\mathrm{Gen}_{@})$ from $\vdash\varphi$ to $\vdash @_i\varphi$. Indeed, replacement applied to $i\land\top\leftrightarrow i$, together with $(\mathrm{Conj}_{@})$ and $(\mathrm{Ref}_{@})$, gives $\vdash @_i\top$. If $\vdash\varphi$, replacement applied to $\varphi\leftrightarrow\top$ then gives $\vdash @_i\varphi$.

We define the following set of ground formulas:
\[
\mathsf{GFm}:=
\{@_ip:i\in\Omega,\ p\in P\}
\cup
\{@_i\langle r\rangle j:i,j\in\Omega,\ r\in G\}.
\]
In the target of the translation, each member of $\mathsf{GFm}$ is treated as a single propositional variable. Thus the @ and modal symbols occurring inside a ground formula are not recursively interpreted by the target propositional language.

For every evaluation nominal $i\in\Omega$, define $\tau_i$ recursively by
\begin{align*}
\tau_i(p)&=@_ip,
&
\tau_i(j)&=
\begin{cases}
\top,&i=j,\\
\bot,&i\neq j,
\end{cases}
\\
\tau_i(\neg\varphi)&=\neg\tau_i(\varphi),
&
\tau_i(\varphi\land\psi)&=
\tau_i(\varphi)\land\tau_i(\psi),
\\
\tau_i(@_j\varphi)&=\tau_j(\varphi),
&
\tau_i(\langle r\rangle\varphi)&=
\bigvee_{j\in\Omega}
\bigl(
@_i\langle r\rangle j
\land\tau_j(\varphi)
\bigr).
\end{align*}

For $\varphi\in\Lang_0$, put
$\rho_\varphi:=\bigvee_{i\in\Omega}(i\land\tau_i(\varphi))$.
Given a static network model $M$ with $V(i)=\{a_i\}$, let $v_M$
be the propositional valuation on $\mathsf{GFm}$ determined by
$v_M(@_ip)=1$ iff $M,a_i\models p$, and
$v_M(@_i\langle r\rangle j)=1$ iff $a_i\sim_r a_j$.

\begin{lemma}
\label{lem:grounding}
For every $\varphi\in\Lang_0$, the reduct $\rho_\varphi$ is a Boolean
combination of nominals and members of $\mathsf{GFm}$.
For every static network model $M$ and every $i\in\Omega$,
\begin{equation}
M,a_i\models\varphi
\quad\text{iff}\quad
v_M\models\tau_i(\varphi).
\label{eq:groundingtruth}
\end{equation}
Moreover,
\begin{equation}
\vdash\varphi\leftrightarrow\rho_\varphi.
\label{eq:provablegrounding}
\end{equation}
Both $\rho_\varphi$ and a derivation of this equivalence are
effectively computable from $\varphi$.
\end{lemma}

\begin{proof}
We first prove \eqref{eq:groundingtruth} by induction on
$\varphi$, simultaneously for all $i\in\Omega$.
The atomic and nominal cases follow from the definitions of
$v_M$ and $\tau_i$. The Boolean and $@$ cases follow from
the induction hypothesis.

Suppose that $\varphi=\langle r\rangle\psi$.
Since every agent of $M$ is named, we have
\[
\begin{aligned}
M,a_i\models\langle r\rangle\psi
&\;\text{iff}\; \exists j\in\Omega\,
  \bigl(a_i\sim_r a_j \text{ and } M,a_j\models\psi\bigr)\\
&\;\text{iff}\; v_M\models
  \bigvee_{j\in\Omega}
  \bigl(@_i\langle r\rangle j\land\tau_j(\psi)\bigr)\\
&\;\text{iff}\; v_M\models\tau_i(\langle r\rangle\psi).
\end{aligned}
\]
The second equivalence uses the induction hypothesis and the
definition of $v_M$; the last uses the definition of $\tau_i$.
This proves \eqref{eq:groundingtruth}.

To prove \eqref{eq:provablegrounding}, we first show that
\[
\vdash @_i\varphi\leftrightarrow\tau_i(\varphi)
\qquad\text{for every }i\in\Omega.
\tag{*}
\]
We again proceed by induction on $\varphi$, simultaneously
for all $i\in\Omega$, treating members of $\mathsf{GFm}$ as
propositional atoms. We give only the modal case.
By $(\mathrm{Paste}_r)$,
$\vdash @_i\langle r\rangle\psi
\leftrightarrow
\bigvee_{j\in\Omega}
\bigl(@_i\langle r\rangle j\land @_j\psi\bigr)$.

By the induction hypothesis,
$\vdash @_j\psi\leftrightarrow\tau_j(\psi)$ for every
$j\in\Omega$. Replacing provably equivalent subformulas and
using the definition of $\tau_i$, we obtain
$\vdash @_i\langle r\rangle\psi
\leftrightarrow\tau_i(\langle r\rangle\psi)$. This completes the induction.

Finally, $(\mathrm{Name})$, $(\mathrm{Locality})$, and
propositional reasoning give $\vdash \varphi\leftrightarrow
\bigvee_{i\in\Omega}\bigl(i\land @_i\varphi\bigr)$. Using $(*)$ and the definition
$\rho_\varphi=\bigvee_{i\in\Omega}(i\land\tau_i(\varphi))$,
we obtain $\vdash\varphi\leftrightarrow\rho_\varphi$,
as required.
\end{proof}

To encode the frame conditions propositionally, write $R^r_{ij}:=@_i\langle r\rangle j$.
Let $C_{\mathrm{t}}$ be the following formula, where nominal and threshold indices range over $\Omega$ and $G$,
\begin{align*}
C_{\mathrm{t}}:={}&
(\bigvee_i i)
\land\bigwedge_{i\neq j}\neg(i\land j)
\land\bigwedge_{i,r}R^r_{ii}
\land\bigwedge_{i,j,r}(R^r_{ij}\leftrightarrow R^r_{ji})
\\[-2pt]
&\land\bigwedge_{\substack{i,j,r,s\\r\leq s}}
(R^r_{ij}\to R^s_{ij})
\land\bigwedge_{i,j,k,r,s}
\bigl((R^r_{ij}\land R^s_{jk})\to R^{r\oplus s}_{ik}\bigr)
\land\bigwedge_{i,j}R^1_{ij}.
\end{align*}
The naming conjuncts follow from $(\mathrm{Name})$, $(\mathrm{Dist})$, and $(\mathrm{Locality})$. The remaining conjuncts are the ground forms of $(T_r)$, $(B_r)$, $(\mathrm{Mon})$, $(\mathrm{Tri})$, and $(\mathrm{Bd})$. Using $(\mathrm{Gen}_{@})$ and Lemma~\ref{lem:grounding}, we obtain $\vdash C_{\mathrm{t}}$.

\begin{samepage}
\begin{theorem}[Soundness and strong completeness of $H_0$]
\label{thm:static-completeness}
For every $\Gamma\cup\{\varphi\}\subseteq\Lang_0$,
\[
\Gamma\vdash\varphi
\quad\text{iff}\quad
\Gamma\models\varphi
\]
over static network models.
\end{theorem}

\begin{proof}
For soundness, the frame conditions validate the threshold schemes,
while bijective naming and the satisfaction clause validate the naming
and hybrid schemes. Modus ponens preserves local consequence;
theorem-only necessitation, $(\mathrm{Gen}_{@})$, and replacement
preserve validity, the last by induction on formula contexts.

For strong completeness, suppose $\Gamma\nvdash_{H_0}\varphi$.
By Lemma~\ref{lem:grounding},
$\vdash\chi\leftrightarrow\rho_\chi$ for every $\chi$.
Treat nominals and ground formulas as propositional variables and put
$\Delta:=\{\rho_\gamma:\gamma\in\Gamma\}
\cup\{C_{\mathrm t},\neg\rho_\varphi\}$.
This set is propositionally satisfiable: otherwise compactness gives
a finite unsatisfiable subset, and propositional completeness,
$\vdash C_{\mathrm t}$, and the grounding equivalences derive
$\varphi$ from finitely many members of $\Gamma$, a contradiction.

Let $u$ satisfy $\Delta$, extending it arbitrarily to unused
variables. The naming conjuncts select a unique $k\in\Omega$ with
$u(k)=1$. Choose a bijection $\nu:\Omega\to A$ and set
$V(i)=\{\nu(i)\}$,
$\nu(i)\in V(p)$ iff $u(@_ip)=1$, and
$\nu(i)\sim_r\nu(j)$ iff $u(R^r_{ij})=1$.
The relational conjuncts of $C_{\mathrm t}$ ensure all frame
conditions, so this defines a static network model $M$ whose induced
ground valuation agrees with $u$. Uniqueness of $k$ and
Lemma~\ref{lem:grounding} give, for every $\chi$,
\[
u\models\rho_\chi
\quad\text{iff}\quad u\models\tau_k(\chi)
\quad\text{iff}\quad M,\nu(k)\models\chi.
\]
Hence $M,\nu(k)\models\Gamma$ and
$M,\nu(k)\not\models\varphi$, as required.
The argument permits arbitrary $\Gamma$ and an infinite $P$;
the fixed finite naming set ensures $|A|=|\Omega|$.
\end{proof}
\end{samepage}

\begin{corollary}[Decidability of $\Lang_0$]
\label{cor:static-decidability}
Satisfiability and validity for $\Lang_0$ over static network models
are decidable.
\end{corollary}

\begin{proof}
Given \(\varphi\), we construct \(\rho_\varphi\) and
\(C_{\mathrm t}\). Both formulas are finite, since \(\Omega\) and \(G\)
are finite and only finitely many atoms occur in \(\varphi\), even when
\(P\) is infinite. By the grounding lemma and the preceding countermodel
construction, \(\varphi\) is satisfiable over static network models if
and only if \(C_{\mathrm t}\land\rho_\varphi\) is 
satisfiable. Similarly, \(\varphi\) is valid if and only if
\(C_{\mathrm t}\land\neg\rho_\varphi\) is propositionally
unsatisfiable. Since these are finite propositional formulas, both
questions are decidable by truth-table enumeration.
\end{proof}

\section{Dynamic characterization of filter bubbles}

A static network model represents the attitudes of the agents and the threshold relations among them at a given stage. The selective delivery of information may change those attitudes, thereby preserving or dissolving an existing filter bubble or forming a new one. A feed model represents such delivery and its possible attitudinal effects. Applying it to a static network model gives an updated model that records the resulting attitudes and relations.

\subsection{Attitude profiles and network models}

To make the dependence on attitudes explicit, let
$V:P\cup\Omega\to\wp(A)$ be a valuation. The \emph{attitude profile}
induced by $V$ at $a\in A$ is
\begin{equation}
S_V(a):=\{p\in P:a\in V(p)\}.
\label{eq:profile}
\end{equation}
If $V$ is the valuation of a model $M$, write $S_M(a)$ for $S_V(a)$.

The comparison between agents is mediated by their attitude profiles.
Fix a normalized pseudometric
\[
\delta:\wp(P)\times\wp(P)\longrightarrow[0,1]
\]
such that, for all $S,T,U\subseteq P$,
\[
\delta(S,S)=0,\qquad
\delta(S,T)=\delta(T,S),\qquad
\delta(S,U)\leq\delta(S,T)+\delta(T,U).
\]
For agents $a,b\in A$, the value
$\delta(S_V(a),S_V(b))$ is their profile-based distance. 

A profile-based distance need not belong to $G$, whereas the
 $\langle r\rangle\varphi$ modalities are indexed by thresholds in $G$. We adapt the
distance to this scale by rounding it upward to the least grid value
not smaller than it. For $x\in[0,1]$, define
\begin{equation}
\lceil x\rceil_G
:=\min\{r\in G:x\leq r\}
=\frac{\lceil mx\rceil}{m}
\qquad(x\in[0,1]).
\label{eq:grid-ceiling}
\end{equation}
For every $x\in[0,1]$ and $r\in G$,
$\lceil x\rceil_G\leq r$ iff $x\leq r$. Hence upward rounding adapts
profile comparison to the scale $G$ without changing any threshold
comparison expressible in the language.

We now make this profile-based comparison the relational component of
the model.

\begin{definition}[Network model]
\label{def:network-model}
A \emph{network model}, relative to $\delta$, is a tuple
\[
N=\bigl(A,(\sim_r)_{r\in G},V\bigr),
\]
where $A$ is finite and nonempty. For every $i\in\Omega$, $V(i)$ is a
singleton, and the map sending $i$ to the unique member of $V(i)$ is a
bijection from $\Omega$ onto $A$.
For all $a,b\in A$ and $r\in G$, the threshold relations are defined
from the agents' profiles by
\begin{equation}
a\sim_r b
\quad\text{iff}\quad
\left\lceil
\delta\bigl(S_N(a),S_N(b)\bigr)
\right\rceil_G
\leq r.
\label{eq:network-relation}
\end{equation}
\end{definition}


\begin{proposition}
\label{prop:network-static}
Every network model is a static network model.
\end{proposition}

\begin{proof}
The nominal valuation satisfies the naming condition required of a
static model. Reflexivity and symmetry follow from the corresponding
properties of $\delta$, boundedness from normalization, and nesting
from the order on $G$.

For triangle closure, suppose that $a\sim_r b$ and $b\sim_s c$.
Since $x\leq\lceil x\rceil_G$ for every $x\in[0,1]$,
equation~\eqref{eq:network-relation} gives
$\delta(S_N(a),S_N(b))\leq r$ and
$\delta(S_N(b),S_N(c))\leq s$. The triangle inequality and
normalization yield
\[
\delta(S_N(a),S_N(c))
\leq\min\{1,r+s\}=r\oplus s.
\]
Since $r\oplus s\in G$, upward rounding gives
$\lceil\delta(S_N(a),S_N(c))\rceil_G\leq r\oplus s$.
Thus $a\sim_{r\oplus s}c$.
\end{proof}

Network models are therefore $\delta$-generated static models. The
inclusion is generally proper, since a static model may interpret its
primitive threshold relations independently of its attitude
valuation. 

\subsection{Feed models and updates}


Preconditions directly specify the platform's allocation of information feeds, while postconditions specify how the selected event determines an agent's updated attitudes based on the state of the network model. This separates the platform's allocation rule from the selected event's effect on the agent's attitudes.


\begin{definition}[Feed model]
A feed model is a triple
\[
\Act=(E,\pre,\post),
\]
where $E$ is a finite nonempty set of events and
\[
\pre:E\longrightarrow\Lang_0,
\qquad
\post:E\times P\longrightarrow\Lang_0
\]
are total functions satisfying
\begin{align*}
\models \bigvee_{e\in E}\pre(e),
\qquad
\models \neg\bigl(\pre(e)\land\pre(f)\bigr)
\quad(e,f\in E,\ e\ne f).
\end{align*}
\end{definition}

For an input network model $M$, the event $e$ is selected at $a$
exactly when
\[
M,a\models\pre(e).
\]
If $e$ is selected at $a$, the truth value of $\post(e,p)$ at $a$ in
$M$ becomes the new value of $p$. All postconditions are evaluated in
the input model before any new value is assigned, so the update of the
attitude valuation is simultaneous.

The two validity conditions make the preconditions jointly exhaustive
and pairwise incompatible over network models. Consequently, exactly
one event is selected at each agent, although the same event may be
selected at several agents. Since no accessibility relations are
imposed on $E$, the present framework does not represent uncertainty
about which event is selected
\parencite{baltag1999,ditmarsch2007}.

Let $M$ be a network model and $\Act$ a feed model. For every $a\in A$,
write $e_a$ for the unique event such that
$M,a\models\pre(e_a)$.

\begin{definition}
The update of $M$ by $\Act$ is
\[
M\otimes\Act=
\bigl(A,(\sim_r^\otimes)_{r\in G},V^\otimes\bigr).
\]
For $p\in P$, the updated propositional valuation is
\begin{equation}
a\in V^\otimes(p)
\quad\text{iff}\quad
M,a\models\post(e_a,p),
\label{eq:updatevaluation}
\end{equation}
for every nominal $i\in\Omega$,
\begin{equation}
V^\otimes(i)=V(i),
\label{eq:updatenominal}
\end{equation}
and for all $a,b\in A$ and $r\in G$,
\begin{equation}
a\sim_r^\otimes b
\quad\text{iff}\quad
\left\lceil\delta\bigl(S_{V^\otimes}(a),S_{V^\otimes}(b)\bigr)
\right\rceil_G\leq r.
\label{eq:updaterelation}
\end{equation}
\end{definition}

\begin{proposition}
\label{prop:update-closure}
If $M$ is a network model and $\Act$ is a feed model, then
$M\otimes\Act$ is a network model.
\end{proposition}

\begin{proof}
Exhaustivity and incompatibility make each $e_a$ unique, so
equation~\eqref{eq:updatevaluation} defines a total valuation on $P$.
Equation~\eqref{eq:updatenominal} preserves the naming bijection.
Finally, equation~\eqref{eq:updaterelation} is precisely the generating
condition~\eqref{eq:network-relation} for $V^\otimes$. Hence
$M\otimes\Act$ is a network model.
\end{proof}



\section{The dynamic logic of filter bubbles}

\subsection{Language and semantics}

Let $\Sigma$ be a collection of feed models, $P$ be a nonempty finite set of propositional atoms, and $\Omega$ be a finite set of nominals, with $P\cap\Omega=\varnothing$ and $|\Omega|=|A|$. The dynamic language
$\Lang_D$ is generated by
\begin{equation}
\varphi::=p\mid i\mid\neg\varphi\mid(\varphi\land\psi)
\mid\langle r\rangle\varphi\mid@_i\varphi\mid[\Act]\varphi
\end{equation}
where $p\in P$, $i\in\Omega$, $r\in G$, and $\Act\in\Sigma$.
Formulas are evaluated at agents of network models. The static clauses
are inherited from Section~2.2, and the dynamic clause is
\begin{equation}
M,a\models[\Act]\varphi
\quad\text{iff}\quad
M\otimes\Act,a\models\varphi
\label{eq:dynamicsemantics}
\end{equation}
Let $\Gamma\cup\{\varphi\}$ be a set formulas in $\Lang_D$, and we use
$\Gamma\models\varphi$ to denote local
consequence relation.

\subsection{Profile formulas  and updated relation}

For
$i\in\Omega$ and $S\subseteq P$, we define the profile formula $\Prof_S^{\Act}(i)$  of the agent named by $i$ after the update through the feed mode $\Act=(E,\pre,\post)$. 
\begin{equation}
\Prof_S^{\Act}(i):=
@_i\bigwedge_{e\in E}
\left(
\pre(e)\rightarrow
\left(
\bigwedge_{p\in S}\post(e,p)
\land
\bigwedge_{p\in P\setminus S}\neg\post(e,p)
\right)
\right)
\label{eq:profileformula}
\end{equation}
As usual, an empty conjunction is identified with $\top$. Because the preconditions select a unique event, this formula requires
precisely the atoms in $S$, and no others, to be true after the update
at the agent named by $i$.

For $i,j\in\Omega$ and $r\in G$, we define the formula $\Sim_r^{\Act}(i,j)$ of updated relation between the agents named by $i$ and $j$ after the update through the feed mode $\Act=(E,\pre,\post)$ in the following.
\begin{equation}
\Sim_r^{\Act}(i,j)
:=
\bigvee_{\substack{S,T\subseteq P\\
\lceil\delta(S,T)\rceil_G\leq r}}
\bigl(\Prof_S^{\Act}(i)\land\Prof_T^{\Act}(j)\bigr)
\label{eq:updatedsimilarity}
\end{equation}

\begin{lemma}
\label{lem:updated-similarity}
For every network model $M$, feed model $\Act$, nominals $i,j\in\Omega$, and threshold $r\in G$, if $V(i)=\{a\}$ and $V(j)=\{b\}$, then
\[
M\models\Sim_r^{\Act}(i,j)
\quad\text{iff}\quad
a\sim_r^\otimes b.
\]
\end{lemma}

\begin{proof}
Let $e_a$ be the unique event selected at $a$. Since $V(i)=\{a\}$, then we have
\[
M\models\Prof_S^{\Act}(i)
\quad\text{iff}\quad
M,a\models\left(
\bigwedge_{p\in S}\post(e_a,p)
\land
\bigwedge_{p\in P\setminus S}\neg\post(e_a,p)
\right).
\]
By equation~\eqref{eq:updatevaluation}, the first conjunction gives
$S\subseteq S_{V^\otimes}(a)$ and the second gives
$S_{V^\otimes}(a)\subseteq S$. Hence
\[
M\models\Prof_S^{\Act}(i)
\quad\text{iff}\quad
S=S_{V^\otimes}(a).
\]
The same argument applies to $j$ and $b$. Consequently, only the pair
$(S_{V^\otimes}(a),S_{V^\otimes}(b))$ can witness
equation~\eqref{eq:updatedsimilarity}.
Equation~\eqref{eq:updaterelation} now gives the required equivalence.
\end{proof}
\subsection{Filter bubble formation, persistence, and dissolution}

The static formula $B_I^{\varepsilon,\eta}$ describes a bubble-shaped
configuration at some stage. To classify the effect of a fixed feed
model, we compare its truth before and after the update. Fix
$\Act\in\Sigma$, $\varnothing\neq I\subsetneq\Omega$, and
$\varepsilon,\eta\in G$ with $\varepsilon<\eta$, and define
\begin{align}
\FBForm_{\Act}(I;\varepsilon,\eta)
&:=\neg B_I^{\varepsilon,\eta}
\land[\Act]B_I^{\varepsilon,\eta},
\label{eq:formation}\\
\FBPres_{\Act}(I;\varepsilon,\eta)
&:=B_I^{\varepsilon,\eta}
\land[\Act]B_I^{\varepsilon,\eta},
\label{eq:persistence}\\
\FBDiss_{\Act}(I;\varepsilon,\eta)
&:=B_I^{\varepsilon,\eta}
\land[\Act]\neg B_I^{\varepsilon,\eta}.
\label{eq:dissolution}
\end{align}
These formulas distinguish formation, persistence, and dissolution as
three transition patterns. Formation
and dissolution record a switch after the update, whereas persistence
records its preservation. All three notions can already be expressed in static language, we show the key proposition as follows.

\begin{proposition}
\label{prop:output-bubble}
For every network model $M$,
$\Act\in\Sigma$, $\varnothing\neq I\subsetneq\Omega$, and
$\varepsilon,\eta\in G$ with $\varepsilon<\eta$,
\begin{align*}
M\models[\Act]B_I^{\varepsilon,\eta}
\quad\text{iff}\quad
M\models{}&
\bigwedge_{i,j\in I}\Sim_\varepsilon^{\Act}(i,j)
\land
\bigwedge_{\substack{i\in I\\j\in\Omega\setminus I}}
\neg\Sim_\eta^{\Act}(i,j).
\end{align*}
\end{proposition}

\begin{proof}
For each $k\in\Omega$, let $V(k)=\{a_k\}$ and put
$M^\otimes:=M\otimes\Act$. By equation~\eqref{eq:updatenominal}, every
nominal $k$ still names $a_k$ in $M^\otimes$. We have
\begingroup
\setlength{\jot}{0pt}
\begin{align*}
& M\models[\Act]B_I^{\varepsilon,\eta}
\\
\text{iff}\quad
& M^\otimes\models B_I^{\varepsilon,\eta}
\\
\text{iff}\quad
& M^\otimes\models
  \bigwedge_{i\in I}@_i
  \left(\bigwedge_{j\in I}\langle\varepsilon\rangle j
  \land[\eta]f_I\right)
\\
\text{iff}\quad
& M^\otimes\models
  \bigwedge_{i\in I}
  \left(\left(@_i\bigwedge_{j\in I}\langle\varepsilon\rangle j\right)
  \land @_i[\eta]f_I\right)
\\
\text{iff}\quad
& M^\otimes\models
  \bigwedge_{i\in I}@_i\bigwedge_{j\in I}\langle\varepsilon\rangle j
  \land\bigwedge_{i\in I}@_i[\eta]f_I
\\
\text{iff}\quad
& M^\otimes\models
  \bigwedge_{i\in I}\bigwedge_{j\in I}@_i\langle\varepsilon\rangle j
  \land\bigwedge_{i\in I}@_i\neg\langle\eta\rangle\neg f_I
\\
\text{iff}\quad
& M^\otimes\models
  \bigwedge_{i\in I}\bigwedge_{j\in I}@_i\langle\varepsilon\rangle j
  \land\bigwedge_{i\in I}@_i\neg\langle\eta\rangle
  \left(\neg\bigvee_{j\in I}j\right)
\\
\intertext{Since every agent is named by exactly one nominal in
$\Omega$, we can express the complement condition using the names
in $\Omega\setminus I$:}
\text{iff}\quad
& M^\otimes\models
  \bigwedge_{i\in I}\bigwedge_{j\in I}@_i\langle\varepsilon\rangle j
  \land\bigwedge_{i\in I}@_i\neg\langle\eta\rangle
  \left(\bigvee_{j\in\Omega\setminus I}j\right)
\\
\text{iff}\quad
& M^\otimes\models
  \bigwedge_{i\in I}\bigwedge_{j\in I}@_i\langle\varepsilon\rangle j
  \land\bigwedge_{i\in I}\bigwedge_{j\in\Omega\setminus I}
  \neg @_i\langle\eta\rangle j
\\
\text{iff}\quad
& a_i\sim_\varepsilon^\otimes a_j
  \quad\text{for all }i,j\in I,
\\[-2pt]
& \text{and}\quad a_i\not\sim_\eta^\otimes a_j
  \quad\text{for all }i\in I,\ j\in\Omega\setminus I
\\
\intertext{Lemma~\ref{lem:updated-similarity} identifies these updated
relations with specified formulas:}
\text{iff}\quad
& M\models\Sim_\varepsilon^{\Act}(i,j)
  \quad\text{for all }i,j\in I,
\\[-2pt]
& \text{and}\quad M\not\models\Sim_\eta^{\Act}(i,j)
  \quad\text{for all }i\in I,\ j\in\Omega\setminus I
\\
\text{iff}\quad
& M\models
  \bigwedge_{i,j\in I}\Sim_\varepsilon^{\Act}(i,j)
  \land\bigwedge_{\substack{i\in I\\j\in\Omega\setminus I}}
  \neg\Sim_\eta^{\Act}(i,j).\tag*{\qedhere}
\end{align*}
\endgroup
\end{proof}
The following example illustrates formation through
 some feed model. Let
$A=\{a_1,\ldots,a_6\}$, $\Omega=\{i_1,\ldots,i_6\}$, and
$V(i_t)=\{a_t\}$ for $1\leq t\leq6$. Take
$P=\{p_1,p_2,p_3\}$, $G=\{0,\frac13,\frac23,1\}$, and
$I=\{i_1,i_2,i_3\}$.
For this example, take $\delta=\delta_{\mathrm H}$, where
$\delta_{\mathrm H}$ is the normalized Hamming distance
\begin{equation}
\delta_{\mathrm H}(S,T)
:=
\frac{|S\setminus T|+|T\setminus S|}{|P|}
\qquad(S,T\subseteq P).
\label{eq:example-hamming}
\end{equation}
This is a normalized metric and therefore satisfies the assumptions on
$\delta$.
In the order $p_1,p_2,p_3$, assign to $a_1,\ldots,a_6$ the profiles
\[
100,\quad010,\quad001,\quad110,\quad101,\quad011,
\]
Let $M$ be the initial network model for $\delta_{\mathrm H}$. Since
$|P|=m=3$, every value of
$\delta_{\mathrm H}$ belongs to $G$, so upward rounding has no effect, and
$M\not\models B_I^{\frac13,\frac23}$ because
$d_M(a_1,a_2)=\frac23$.

Let $\Act=(E,\pre,\post)$, where $E=\{e_0,e_1,e_2,e_3\}$ and
\[
\pre(e_0):=i_1,\qquad
\pre(e_1):=i_2,\qquad
\pre(e_2):=i_3,\qquad
\pre(e_3):=i_4\lor i_5\lor i_6.
\]
The postconditions are given in the following table.
\[
\begin{array}{c|ccc}
 &p_1&p_2&p_3\\ \hline
e_0&\top&\bot&p_3\\
e_1&p_2&\neg p_2&p_1\land p_3\\
e_2&p_1\lor p_3&p_1\land p_2&\neg p_3\\
e_3&p_1\land p_2\land p_3
   &p_1\lor p_2\lor p_3
   &(p_1\land p_2)\lor(p_1\land p_3)\lor(p_2\land p_3)
\end{array}
\]

\begin{figure}[htb]
\centering
\begin{tikzpicture}[
  agent/.style={circle,draw,minimum size=8mm,inner sep=1pt,font=\scriptsize},
  member/.style={agent,fill=black!10},
  outsider/.style={agent,fill=white},
  group/.style={draw,dashed,rounded corners,inner sep=4pt},
  panel/.style={font=\small\bfseries},
  note/.style={font=\scriptsize}
]
  \node[panel] at (1.10,2.20) {Before update};
  \node[member]  (b1) at (0.00,0.95) {$\begin{array}{c}a_1\\100\end{array}$};
  \node[member]  (b2) at (1.10,0.95) {$\begin{array}{c}a_2\\010\end{array}$};
  \node[member]  (b3) at (2.20,0.95) {$\begin{array}{c}a_3\\001\end{array}$};
  \node[outsider](b4) at (0.00,-1.10) {$\begin{array}{c}a_4\\110\end{array}$};
  \node[outsider](b5) at (1.10,-1.10) {$\begin{array}{c}a_5\\101\end{array}$};
  \node[outsider](b6) at (2.20,-1.10) {$\begin{array}{c}a_6\\011\end{array}$};
  \node[group,fit=(b1)(b2)(b3)] {};
  \node[group,fit=(b4)(b5)(b6)] {};
  \node[note] at (1.10,-0.08) {$d_M(a_1,a_2)=\frac23>0$};

  \draw[-{Latex[length=2mm]},thick] (3.05,-0.08) -- node[above,note] {$\Act$} (4.15,-0.08);

  \node[panel] at (6.10,2.20) {After update};
  \node[member]  (a1) at (5.00,0.95) {$\begin{array}{c}a_1\\100\end{array}$};
  \node[member]  (a2) at (6.10,0.95) {$\begin{array}{c}a_2\\100\end{array}$};
  \node[member]  (a3) at (7.20,0.95) {$\begin{array}{c}a_3\\100\end{array}$};
  \node[outsider](a4) at (5.00,-1.10) {$\begin{array}{c}a_4\\011\end{array}$};
  \node[outsider](a5) at (6.10,-1.10) {$\begin{array}{c}a_5\\011\end{array}$};
  \node[outsider](a6) at (7.20,-1.10) {$\begin{array}{c}a_6\\011\end{array}$};
  \node[group,fit=(a1)(a2)(a3)] {};
  \node[group,fit=(a4)(a5)(a6)] {};
  \node[note] at (6.10,-0.08) {$d_{\mathrm{in}}=0,\quad d_{\mathrm{cross}}=1$};
\end{tikzpicture}
\caption{Formation of a filter bubble under $\Act$}
\label{fig:six-agent-update}
\end{figure}

Evaluation of the postconditions gives
$S_{M\otimes\Act}(a_k)=\{p_1\}$ for $1\leq k\leq3$ and
$S_{M\otimes\Act}(a_k)=\{p_2,p_3\}$ for $4\leq k\leq6$. Hence all
distances within $F_I^{M\otimes\Act}$ are $0$, whereas all
cross-boundary distances are $1$. Since the initial model fails
$B_I^{\frac13,\frac23}$ and the updated model satisfies it,
\[
M\models\FBForm_{\Act}(I;\frac13,\frac23).
\]

\subsection{Axiomatic system $H_D$}

For $S\subseteq P$, let
\begin{equation}
\pi_S:=
\bigwedge_{p\in S}p
\land
\bigwedge_{p\in P\setminus S}\neg p,
\label{eq:exact-profile}
\end{equation}
where an empty conjunction is $\top$. Thus
$N,a\models\pi_S$ iff $S_N(a)=S$.

For the fixed effectively presented pseudometric $\delta$, let
$H_{\mathsf D}$ extend $H_0$ by replacement of provable equivalents in
arbitrary $\Lang_D$-contexts and by the following schemes, where
$r\in G$ and $\Act\in\Sigma$:

\begingroup
\makeatletter
\newcommand{\HDSchemeLabel}[2]{%
  \phantomsection\def\@currentlabel{{#1}}\label{#2}}
\makeatother
\begin{center}
\begin{tabular}{@{}ll@{}}
\hline
\multicolumn{2}{@{}l}{\textit{Network scheme}} \\[2pt]

$\displaystyle
\langle r\rangle\varphi
\leftrightarrow
\bigvee_{\substack{S,T\subseteq P\\
\lceil\delta(S,T)\rceil_G\leq r}}
\left(
\pi_S\land
\bigvee_{j\in\Omega}@_j(\pi_T\land\varphi)
\right)$
& $(\mathrm{Net}_r)$\HDSchemeLabel{Net$_r$}{ax:network-base} \\[4pt]
\hline
\multicolumn{2}{@{}l}{\textit{Reduction schemes}} \\[2pt]

$\displaystyle
[\Act]p\leftrightarrow
\bigwedge_{e\in E}\bigl(\pre(e)\to\post(e,p)\bigr)$
& $(\mathrm{RA}\mbox{-}\mathrm{Atom})$\HDSchemeLabel{RA-Atom}{ra:atom} \\

$\displaystyle [\Act]i\leftrightarrow i$
& $(\mathrm{RA}\mbox{-}\mathrm{Nom})$\HDSchemeLabel{RA-Nom}{ra:nom} \\

$\displaystyle [\Act]\neg\varphi\leftrightarrow\neg[\Act]\varphi$
& $(\mathrm{RA}\mbox{-}\mathrm{Neg})$\HDSchemeLabel{RA-Neg}{ra:neg} \\

$\displaystyle
[\Act](\varphi\land\psi)\leftrightarrow([\Act]\varphi\land[\Act]\psi)$
& $(\mathrm{RA}\mbox{-}\mathrm{Conj})$\HDSchemeLabel{RA-Conj}{ra:conj} \\

$\displaystyle [\Act]@_i\varphi\leftrightarrow @_i[\Act]\varphi$
& $(\mathrm{RA}\mbox{-}@)$\HDSchemeLabel{RA-@}{ra:at} \\

$\displaystyle
[\Act]\langle r\rangle\varphi\leftrightarrow
\bigvee_{i,j\in\Omega}
\bigl(i\land\Sim_r^{\Act}(i,j)\land @_j[\Act]\varphi\bigr)$
& $(\mathrm{RA}\mbox{-}\mathbb{T}_r)$\HDSchemeLabel{RA-$\mathbb{T}_r$}{ra:sim} \\[4pt]
\hline
\end{tabular}
\end{center}
\endgroup

The scheme $(\mathrm{Net}_r)$ connects the profile valuation with the
threshold relations generated by equation~\eqref{eq:network-relation}.
The atomic reduction uses the unique selected event and the
simultaneous evaluation of its postconditions. The nominal and
satisfaction reductions follow from preservation of the nominal
valuation, and the Boolean clauses follow from the total deterministic
update semantics. Finally, $(\mathrm{RA}\mbox{-}\mathbb{T}_r)$ uses
$\Sim_r^{\Act}(i,j)$ to describe the reconstructed relation in
the static language; Lemma~\ref{lem:updated-similarity} supplies the
required semantic equivalence.

Let $\Lang_{@}$ be the threshold-free fragment of $\Lang_0$.

\begin{lemma}[Dynamic reduction]
\label{lem:dynamic-reduction}
There is a translation
$\operatorname{red}:\Lang_D\to\Lang_{@}$ such that
$\vdash_{H_{\mathsf D}}\varphi\leftrightarrow\operatorname{red}(\varphi)$
for every $\varphi\in\Lang_D$.
Given descriptions of the feed models in $\varphi$ and a procedure
deciding the grid comparisons $\lceil\delta(S,T)\rceil_G\leq r$,
both $\operatorname{red}(\varphi)$ and a derivation of this equivalence
are effectively computable from $\varphi$.
\end{lemma}

\begin{proof}
First eliminate updates. Define $R_{\Act}:\Lang_0\to\Lang_0$ by the
matching reduction scheme, replacing residual $[\Act]\psi$ on its
right-hand side by $R_{\Act}(\psi)$. Recursion follows proper
subformulas; preconditions, postconditions, and $\Sim$-formulas are
static. Induction and replacement give
$\vdash_{H_{\mathsf D}}[\Act]\chi\leftrightarrow R_{\Act}(\chi)$.
Define $U:\Lang_D\to\Lang_0$ homomorphically on static constructors,
with $U([\Act]\psi)=R_{\Act}(U(\psi))$.
A second structural induction gives
$\vdash_{H_{\mathsf D}}\varphi\leftrightarrow U(\varphi)$.

Next eliminate thresholds. Define $T:\Lang_0\to\Lang_{@}$
homomorphically except that its $\langle r\rangle\psi$-clause is the
right-hand side of $(\mathrm{Net}_r)$ with $T(\psi)$ substituted for
$\varphi$. Induction using this scheme and replacement gives
$\vdash_{H_{\mathsf D}}\chi\leftrightarrow T(\chi)$.
Thus $\operatorname{red}=T\circ U$ has the required target and
provable equivalence.

Each recursion terminates by structural induction. The finite sets
$P$, $\Omega$, $G$, and event sets make all expansions finite; the assumed
descriptions and grid comparisons make them effective. Recording the
axiom, replacement, and propositional steps computes the derivation.
\end{proof}

We now use this reduction to obtain the dynamic counterparts of Theorem~\ref{thm:static-completeness} and Corollary~\ref{cor:static-decidability}. 

\begin{theorem}[Soundness and strong completeness of $H_{\mathsf D}$]
\label{thm:dynamic-completeness}
For the fixed profile pseudometric $\delta$ and collection $\Sigma$
of feed models, every $\Gamma\cup\{\varphi\}\subseteq\Lang_D$
satisfies
\[
\Gamma\vdash_{H_{\mathsf D}}\varphi
\quad\text{iff}\quad
\Gamma\models\varphi
\]
over network models.
\end{theorem}

\begin{proof}
For soundness, we detail $(\mathrm{RA}\mbox{-}\mathbb{T}_r)$.
Write $V(j)=\{a_j\}$. Update semantics, preservation of names, and
Lemma~\ref{lem:updated-similarity} give
\[
\begin{aligned}
M,a_i\models[\Act]\langle r\rangle\varphi
&\quad\text{iff}\quad
\exists j\in\Omega\;
\bigl(a_i\sim_r^{\otimes}a_j
\text{ and }M\otimes\Act,a_j\models\varphi\bigr)\\
&\quad\text{iff}\quad
M,a_i\models\bigvee_{j\in\Omega}
\bigl(\Sim_r^{\Act}(i,j)\land@_j[\Act]\varphi\bigr).
\end{aligned}
\]
Since $i$ is the only nominal true at $a_i$, this is the required
scheme. The remaining schemes and rules are sound by static
soundness and the network and update clauses; update closure
justifies replacement inside update contexts.

For strong completeness, suppose
$\Gamma\nvdash_{H_{\mathsf D}}\varphi$ and put
$\operatorname{red}[\Gamma]
:=\{\operatorname{red}(\gamma):\gamma\in\Gamma\}$.
Then $\operatorname{red}[\Gamma]
\nvdash_{H_0}\operatorname{red}(\varphi)$.
Otherwise, its finite derivation uses only finitely many reduced
premises. By Lemma~\ref{lem:dynamic-reduction}, each is derivable
from its original premise in $H_{\mathsf D}$. Since this system
contains $H_0$, reproducing the derivation and using the conclusion's
reduction equivalence gives $\Gamma\vdash_{H_{\mathsf D}}\varphi$,
a contradiction.

Theorem~\ref{thm:static-completeness} therefore supplies a pointed
static model satisfying $\operatorname{red}[\Gamma]$ and falsifying
$\operatorname{red}(\varphi)$. Regenerate its relations
by~\eqref{eq:network-relation}, retaining its domain, valuation, and
evaluation point. The result is a network model. All reducts retain
their truth values: induction uses only atoms, nominals, Boolean
connectives, and satisfaction operators. By soundness and the
reduction equivalences, this network model satisfies $\Gamma$ and
falsifies $\varphi$. Thus strong completeness holds for arbitrary
$\Gamma$.
\end{proof}

\begin{corollary}[Decidability of $\Lang_D$]
\label{cor:dynamic-decidability}
For every finite collection $\Sigma$ of effectively presented feed
models and every profile pseudometric $\delta$ with decidable grid
comparisons, satisfiability and validity for $\Lang_D$ over network
models are decidable.
\end{corollary}

\begin{proof}
Given $\varphi$, compute $\operatorname{red}(\varphi)$ by
Lemma~\ref{lem:dynamic-reduction} using the assumed effective
descriptions and comparisons. The preceding model conversion and
Proposition~\ref{prop:network-static} show that this reduct has the
same satisfiability and validity over static and network models.
Corollary~\ref{cor:static-decidability} decides both questions for
the reduct, and its valid reduction equivalence transfers the
answers to $\varphi$.
\end{proof}
\section{Conclusion}

In this paper, we provided logics for reasoning about filter bubbles, capturing the structural description of a bubble-shaped group separately from its transformation under a personalized feed. We first constructed a static logic that takes threshold relations as primitive and characterizes bubble-shaped configurations through internal cohesion and separation from outsiders. For the dynamic analysis, we introduced network models whose relations are generated from attitude profiles by a fixed pseudometric. A feed model assigned events through preconditions and determined the resulting attitudes through postconditions. The updated relations were then reconstructed from the new profiles by the same comparison rule. The dynamic language distinguished the formation, persistence, and dissolution of a bubble under a specified feed model. We proved that the static and dynamic logics are sound and strongly complete and that their satisfiability and validity problems are decidable.

\printbibliography

\end{document}